\documentclass{cimart}

\usepackage{multirow}

\title{On topological and algebraic structures of categorical random variables}

\authors{Inocencio Ortiz, Santiago G\'{o}mez-Guerrero and Christian Schaerer}

\authorinfo[I. Ortiz]{National University of Asuncion, San Lorenzo, Paraguay}{iortiz@ing.una.py}

\authorinfo[S. G\'{o}mez-Guerrero]{National University of Asuncion, San Lorenzo, Paraguay}{sgomezpy@gmail.com}

\authorinfo[C. Schaerer]{National University of Asuncion, San Lorenzo, Paraguay}{cschaer@pol.una.py}

\abstract{Based on entropy and symmetric uncertainty (SU), we define a metric for categorical random variables and show that this metric can be promoted into an appropriate quotient space of categorical random variables. Moreover, we show that there is a natural commutative monoid structure in the same quotient space, which is compatible with the topology induced by the metric, in the sense that the monoid operation is continuous.}


\acknowledgments{The authors acknowledge Professor Mauricio Poletti for fruitful discussions on the topics of this paper. The authors also acknowledge support from FEEI--CONACYT--PROCIENCIA--PRONII and the ARASY project (No.~ESTR01-23). CES is also affiliated with the Centro de Investigación en Matemática (CIMA), Asunción, Paraguay.}


\VOLUME{34}
\YEAR{2026}
\ISSUE{2}
\NUMBER{9}
\DOI{https://doi.org/10.46298/cm.17035}
\editinfo{December 4, 2025}{March 2, 2026}{Tiago Macedo, Jaqueline Mesquita, Rafael Potrie, and Mariel Saez}

\begin{document}

\section{Introduction}
Symmetric Uncertainty (SU) has been introduced in \cite{CUP:Press_et_al:1988}, and it is defined as: 

\begin{equation} \label{suDefinition}
	SU(X, Y) := 2 \left[ 1 - \frac{H(X, Y)}{H(X) + H(Y)}\right],
\end{equation}
where $X, Y$ are categorical random variables (i.e., random variables taking values in a finite set, not necessarily numerical), $H(\cdot)$ is the Shannon entropy \cite{Shannon-1948}, and $H(\cdot, \cdot)$ is the joint Shannon entropy.

The generalization of SU for the multivariate case, as well as its properties and, in particular, its potential to define a distance metric among categorical random variables have been explored in \cite{Entropy:Gomez-2021,ref-Gomez2023,IS:Sosa-2018, SBMAC:Sosa-2018}. Indeed, \cite{ref-Gomez2023} shows that entropic correlation is an example of similarity.

In this work, we further explore some properties of $SU$ in two directions. On the one hand, we explore some properties of the topology introduced by the metric defined in terms of $SU$ on a given space of categorical random variables. On the other hand, we explore some algebraic properties of a {\it joint} operation defined on the same space where the metric is defined. We conclude that this operation endows the metric space with a commutative monoid structure, and that it is compatible with the metric topology, in the sense that it is a continuous map.

Our work is organized as follows. 

In Section 2, we review some basic concepts from probability and information theory in order to understand the motivation and definition of the symmetric uncertainty. We also discuss the concept of entropy from two perspectives: one related to categorical random variables and another related to partitions of a set. Finally, we conclude this section with some technical lemmas concerning entropy and joint entropy that will be needed later.

In Section 3, we define an appropriate quotient space of categorical random variables and endow it with a metric structure based on entropy and symmetric uncertainty.

In Section 4, we endow the same quotient space of categorical random variables with a monoid structure and establish the compatibility of the topological and algebraic structures.

Finally, in Section 5, we offer some concluding remarks highlighting the potential usefulness of our formalization for statistical practitioners.

\section{Background: Entropic Correlation and Symmetric Uncertainty}
\label{sec:theory}

In this section, we recall the fundamental concepts involved in the definition of $SU$, as well as its main properties, which we will need.

Let us consider a fixed probability space $(D, \Omega_D, P)$, and let $X\colon D\to \Sigma_X$ be a categorical (non-numeric, nominal or qualitative) random variable with possible values in the finite set  $\Sigma_X = \{x_1,\ldots,x_k\}$ and probability mass function $P_X\colon \Sigma_X\to \mathbb{R}$ defined as \(P_X(x_i):= P(X=x_i) := P(X^{-1}(x_i))\). The entropy $H$ of the variable $X$ is a measure of the uncertainty in predicting the value of $X$, or \textit{size of the uncertainty} on $X$, and is defined as \cite{Renyi-1961}:
\begin{equation} \label{defEntropy}
	H(X) := -\sum_{i} P_X(x_i)\log_{2}(P_X(x_i)).
\end{equation}
Any base could be used for the log, but base 2 is used for historical reasons connected to information theory \cite{Shannon-1948}. $H(X)$ can also be interpreted as a measure of the amount of information a discrete random variable $X$ requires for describing its behavior; or more compactly, as the variety inherent to $X$.

Given another categorical random variable $Y\colon D\to \Sigma_Y$, with possible outcomes in the set \(\Sigma_Y=\{y_1, y_2, \ldots, y_m\}\), we may consider the {\it conditional probability}
\[P(x_i|y_j):=\frac{P(X^{-1}(x_i)\cap Y^{-1}(y_j))}{P(Y^{-1}(y_j))}.\]

Then, we can define the conditional entropy
\begin{equation}
	H(X|Y) := -\sum_{j} \left[ P_Y(y_j)\sum_{i}P(x_i|y_j) 
	\log_{2}(P(x_i|y_j)) \right].
\end{equation}
The conditional entropy $H(X|Y)$ quantifies the amount of information needed to describe the outcome of $X$ given that the value of $Y$ is known.

Using entropy and conditional entropy, the concept of {\it joint entropy} is introduced, which for two categorical random variables $X, Y$ is defined by
\begin{equation}\label{def:joint-entropy}
	H(X, Y) = H(Y)+H(X|Y). 
\end{equation}
For more details, we refer the reader to \cite{Bouda}, page 15.

{\it Mutual Information}. Using entropy and conditional entropy introduced above, the mutual information \footnote{In certain literature mutual information is denoted as $I(X;Y)$; see \cite{ref-Fullwood2023}.} $MI(X;Y)$ is introduced \cite{MKP:Quinlan:1993}. Mutual information measures the reduction in uncertainty about the value of $X$ when the value of $Y$ is known, and is defined as:
\begin{equation}\label{infogainDefinition}
	MI(X;Y) := H(X) - H(X|Y).
\end{equation}
Since $MI$ measures how much the information provided by $Y$ makes it easier to predict the value of $X$, it can be used as a {\it measure of correlation}. It should be noted that:
\begin{itemize}
	\item if $X$ and $Y$ are independent then in equation (\ref{infogainDefinition}) $H(X|Y) = H(X)$, and hence $MI(X;Y)=0$. That is, under independence, knowing $Y$ gives no information about the value of $X$; and 
	
	\item if $X$ and $Y$ are fully correlated then $H(X|Y) = 0$ and hence $MI(X;Y) = H(X)$ applying equation (\ref{infogainDefinition}).
\end{itemize}

Mutual information, taken as a measure of correlation, makes it possible to compare correlations with one another. For example, for any random variables $X$, $Z$ and $Y$, $MI(X;Y) > MI(Z;Y)$ means that, when the value of $Y$ is known, the reduction in uncertainty about $X$ is greater than the reduction in uncertainty about $Z$, which suggests that $X$ is more correlated with $Y$ than with $Z$.

By combining equations (\ref{def:joint-entropy}) and (\ref{infogainDefinition}), one readily obtains 
\begin{equation}\label{MIas3entropies}
	MI(X;Y) = H(X) + H(Y) - H(X, Y),
\end{equation}
and since $H(X, Y)$ is symmetric by direct substitution of $X, Y$ into equation \eqref{suDefinition}, it results that $MI(X;Y)$ is symmetric as well, a quite convenient property for a paired measure. A formal proof can be found, for example, as Theorem 1 in \cite{ref-Fullwood2023}.

{\it Symmetric Uncertainty}. On the other hand, as with other informational association measures, MI tends to increase its value when the number of categories in $X$ and/or $Y$ increases, that is, it is biased towards high cardinality features \cite{ML:Quinlan-1986}. Therefore, seeking to compensate for such bias, MI has to be normalized by the sum of entropies of the features. This measure, called {\it Symmetric Uncertainty}~\cite{CUP:Press_et_al:1988} is expressed as:
\begin{equation} \label{suDefinition1}
	SU(X, Y) := 2 \left[ \frac{MI(X;Y)}{H(X) + H(Y)}\right]
\end{equation}
or, combining with equation (\ref{MIas3entropies}),
\begin{equation} \label{suDefinition2}
	SU(X, Y) := 2 \left[1 - \frac{H(X, Y)}{H(X) + H(Y)}\right].
\end{equation}

Revisiting the two bullet points in the previous subsection, note that:
\begin{itemize}
	\item If $X$ and $Y$ are independent, then from $MI(X;Y)=0$ it follows that $SU(X,Y)=0$.
	\item If $X$ and $Y$ are completely correlated, then from \eqref{infogainDefinition} we obtain
	\[
	MI(X;Y)=H(X).
	\]
	By symmetry of mutual information, $MI(Y;X)=H(Y)$, and therefore
	\[
	SU(X,Y)=1.
	\]
\end{itemize}
As a final observation regarding $SU$, we recall from \cite{IS:Sosa-2018} the following results.

\begin{lemma}\label{lem:entropic-ratio} For any two categorical random variables, we have: 
	\begin{enumerate}
		\item[1)] $0\leq SU(X, Y)\leq 1$
		\item[2)] The {\it Entropic Ratio}, defined as
		\begin{equation}
			R(X, Y) := \frac{H(X, Y)}{H(X)+H(Y)},
		\end{equation}
		where $H(X, Y) = H(X) + H(Y|X)$,
		satisfies the inequality
		\[\frac{1}{2}\leq R(X, Y)\leq 1.\]
	\end{enumerate}
\end{lemma}

It will be convenient to relate the concept of entropy of a categorical random variable to that of the entropy of a partition, as follows: let us consider again our fixed probability space $(D, \Omega_D, P)$. A partition of $D$ is a countable collection of disjoint measurable subsets of $D$, whose union has measure $1$. We will denote by ${\cal D}$ the trivial partition, namely ${\cal D} = \{D\}$. Observe that we will only need finite partitions.

Given a partition ${\cal Q}$ of $D$, 
the entropy of ${\cal Q}$ is defined by
\[H({\cal Q}) := -\sum_{Q\in {\cal Q}}P(Q)\log(P(Q)).\]

Given two partitions ${\cal Q}$ and ${\cal R}$ of $D$, we say that ${\cal Q}$ is coarser than ${\cal R}$, (or equivalently that ${\cal R}$ is finer than ${\cal Q}$), denoted ${\cal Q}\leq {\cal R}$ if every element of ${\cal R}$ is contained in some element of ${\cal Q}$. The entropy of ${\cal Q}$ relative to ${\cal R}$ is defined by
\begin{equation}\label{eq:relative-entropy-partition}
	H({\cal Q}|{\cal R}):= -\sum_{Q\in {\cal Q}}\sum_{R\in \cal R}P(Q\cap R)\log\left(\frac{P(Q\cap R)}{P(R)}\right).
\end{equation}
Notice that, for any partition ${\cal X}$ of $D$ we have $H({\cal X}|{\cal D}) = H({\cal X})$.

Finally, given two partitions ${\cal Q}$ and ${\cal R}$, the partition ${\cal Q}\lor{\cal R}$ is defined as the intersection of the elements of ${\cal Q}$ with the elements of  ${\cal R}$. Notice that we have ${\cal Q}\leq {\cal Q}\lor{\cal R}$ as well as ${\cal R}\leq {\cal Q}\lor{\cal R}$.

Regarding these concepts, we have the following technical results.

\begin{lemma}[\!\!{\cite[Lemma 9.1.5]{Book:Viana-Oliveira}}]\label{lemma:relative-entropy-properties}
	Given partitions ${\cal X}, \cal{Y}, {\cal Z}$, with finite entropy, the following hold:
	\begin{enumerate}
		\item[a)] $H({\cal X}\lor {\cal Y}|{\cal Z}) = H({\cal X}|{\cal Z})+H({\cal Y}|{\cal X}\lor {\cal Z})$.
		\item[b)] If ${\cal X}\leq {\cal Y}$, then $H({\cal X}|{\cal Z})\leq H({\cal Y}|{\cal Z})$ and $H(\cal{Z}|{\cal X})\geq H(\cal{Z}|{\cal Y})$. 
		\item[c)] ${\cal X}\leq {\cal Y}$ if and only if $H({\cal X}|{\cal Y}) = 0$.
	\end{enumerate}
\end{lemma}

Notice that, taking ${\cal Z} ={\cal D}$ in item $a)$ we get, for any two partitions ${\cal X}$ and ${\cal Y}$:
\begin{equation}\label{eq:abs-entropy-less-joint-entropy}
	H({\cal X\lor{\cal Y}}) = H({\cal X})+H({\cal Y|X})\implies H({\cal X})\leq H({\cal X}\lor {\cal Y}).
\end{equation}

Also, from ${\cal Y}\leq {\cal Y}\lor{\cal Z}$ and ${\cal Z}\leq {\cal Y}\lor{\cal Z}$ we get (applying item $b)$)
\begin{equation}\label{eq:refinement-inequality}
	H({\cal X}|{\cal Y}\lor{\cal Z})\leq H({\cal X}|{\cal Y}),\quad H({\cal X}|{\cal Y}\lor{\cal Z})\leq H({\cal X}|{\cal Z})
\end{equation}

Now, if \[X:D\to \Sigma_X = \{x_1, x_2, \ldots, x_m\}\] is a categorical random variable defined on $D$, then it defines a partition ${\cal X}$ of $D$ given by the inverse images of $x_i$, i.e., 
\[{\cal X} := \{X^{-1}(x_i), i=1, 2, \ldots, m\}.\]

The basic observation here is that the concept of entropy of a partition, and the entropy of a partition relative to another partition, recover the corresponding concepts regarding categorical random variables. In particular, given two categorical random variables $X$ and $Y$ defined on $D$, we have
\[H(X, Y) = H({\cal X}\lor {\cal Y}).\]

\section{Topological structure of Categorical Variables}\label{sec:topestruc}

In this section we introduce a topology on a set of categorical variables. We will do so by defining a distance metric based on the SU measure. To motivate the definition, we will first explore some examples in order to get an idea of how the SU measure can give us a way to compare categorical variables.

\subsection{Measuring similarity on a dataset: between observations and between variables}
A set of data, or dataset, consists of $m$ observations taken across $n$ variables, or equivalently, \(n\) random variables \(X_j\) evaluated on \(m\) subjects \(q_i\). Let us set \(x_{ij} = X_j(q_i)\), so we can arrange the observed values as shown in Table \ref{tab:adataset}. 

\begin{table}[ht]
	\centering
	\setlength\arrayrulewidth{0.1pt}
	\renewcommand{\arraystretch}{2.25}
	
	\begin{tabular}{|c c c c c c c|}
		\hline
		$X_1$ & $X_2$ & $X_3$ & $X_4$ & $X_5$ & \ldots & $X_n$ \\
		\hline
		$x_{11}$ & $x_{12}$ & $x_{13}$ & $x_{14}$ & $x_{15}$ & ... & $x_{1n}$ \\
		$x_{21}$ & $x_{22}$ & $x_{23}$ & $x_{24}$ & $x_{25}$ & ... & $x_{2n}$ \\
		\vdots & \vdots & \vdots & \vdots & \vdots &   & \vdots \\
		$x_{m1}$ & $x_{m2}$ & $x_{m3}$ & $x_{m4}$ & $x_{m5}$ & ... & $x_{mn}$ \\
		\hline
	\end{tabular}
	\caption{\label{tab:adataset} {A generic dataset with $n$ column variables and $m$ row observations.}}
\end{table}

Looking row-wise at the dataset, any two observations (rows) can be compared to determine their similarity. If two rows are equal throughout all their values, we say that they are in full similarity or equivalently, we say that the distance between these observations is zero. If not all the values on the two rows are equal, several techniques are available to compute distances in this $n$-dimensional space, where a mix of data types is not unusual; see for instance Suárez-Díaz et al. \cite{UG:Suarez-2020}.

We can also look at the same dataset in a column-wise manner \cite{ref-Gomez2023}. In exploratory experiments and in model construction, one tries to establish how the value in one column (the \textit{outcome} or \textit{class}) results from the values in other columns, revealing a \textit{concomitance} or \textit{correlation} between the random variables. 

Is it possible to find similarities between columns? Since each column measures a specific characteristic, and perhaps has been recorded in its own unit of measurement, we rarely expect to find column similarity in the same way that we find similar rows. However, here one can think of similarity in a slightly different way. It is known that columns may relate to each other via correlations. Correlation does not necessarily imply causality, but it is reasonable to view high correlation between two variables $X_i$ and $X_j$ as an indicator of being linked to the same phenomenon. Because of this, we may think that if the correlation between $X_i$ and $X_j$ is high, then the two variables are ``close to each other''; and conversely, low correlation between $X_i$ and $X_j$ sets these variables apart.

We address the problem of finding similarities (or from a complementary viewpoint, finding distances) between columns, variables or features of a database. That is, given any two variables of the dataset, what is the similarity between them and how can we calculate it? 

For categorical variables, the Symmetric Uncertainty measure of correlation computes this kind of similarity by comparing the joint entropy of the pair of variables with their individual entropies. So the answer to the above questions is: Yes, similarity expressed as correlation between two columns is computable as the SU. This similarity is based on variables belonging to a group of intertwined characteristics. The next example illustrates how two columns may be viewed as similar.

\begin{example}
	Students applying for a certain internship had their scores recorded on 5 traits: Neatness, Punctuality, IQ, Attention Type, Creativity, along with a response variable GotHired.    
\end{example}
\begin{table}[ht]
	\centering
	\scriptsize
	\setlength\arrayrulewidth{0.1pt}
	\renewcommand{\arraystretch}{1.5}
	
	\begin{tabular}{|c|c|c|c|c|c|c|}
		\hline
		$Subject$ & $Neatness$ & $Creativity$ & $Punctuality$ & $IQuotient$ & $Attention Type$ & $GotHired$\\
		\hline
		1 & R & D & L & H & A & N \\
		2 & U & S & E & L & A & N \\
		3 & R & D & L & H & A & Y \\
		4 & U & D & E & L & SU & Y \\
		5 & R & I & L & H & SE & N \\
		6 & S & S & O & A & SU & Y \\
		7 & R & D & L & H & SU & Y \\
		8 & S & I & O & A & SE & N \\
		9 & R & S & L & H & A & N \\
		10 & S & D & O & A & A & Y \\
		11 & R & S & L & H & D & N \\
		12 & R & I & L & H & A & N \\
		13 & U & D & E & L & SE & Y \\
		14 & R & D & L & H & D & Y \\
		15 & R & I & L & H & SU & N \\
		16 & U & I & E & L & SU & N \\
		17 & U & D & E & L & SU & Y \\
		18 & S & D & O & A & A & Y \\
		19 & S & S & O & A & SU & N \\
		20 & R & I & L & H & D & N \\
		\hline
	\end{tabular}
	\caption{\label{tab:msu4b}{Personality traits of students applying for an internship. Variables: Neatness (Untidy/Smooth/Refined), Creativity (Divergent/Solver/Imaginative), Punctuality (Erratic/On-time/Late), Intelligence Quotient (Low/Average/High), Attention Type (Selective/Sustained/Divided/Alternating), GotHired (No/Yes)}}
\end{table}

Table \ref{tab:msu4b} presents an example dataset for the personality traits of students applying for an internship.  Computing probabilities from the observed frequency, Table \ref{SU_NeatnessCalculation} shows the calculation of $SU$ for the variables Neatness and Hired based on the three terms in Equation \eqref{suDefinition2}.

\begin{table}[ht]
	\centering
	\begin{tabular}{|c c c|}
		\hline
		Number of cases = 20 & & \\ \hline
		Neatness & $P$(Neatness) & log $P$(Neatness) \\ 
		U        & 0.25        & -2               \\
		S        & 0.25        & -2               \\
		R        & 0.5         & -1               \\
		$H$(Neatness) = 1.5    & &               \\ \hline
		Hired    & $P$(Hired)    & log $P$(Hired)     \\ 
		N        & 0.55        & -0.8625          \\
		Y        & 0.45        & -1.1520          \\
		$H$(Hired) = 0.9928 &          &            \\ \hline
		Neatness, Hired & $P$(Neatness, Hired) & log $P$(Neatness, Hired) \\ 
		U, N      & 0.1         & -3.3219          \\
		U, Y      & 0.15        & -2.7370          \\
		S, N      & 0.1         & -3.3219          \\
		S, Y      & 0.15        & -2.7370          \\
		R, N      & 0.35        & -1.5146          \\
		R, Y      & 0.15        & -2.7370          \\
		$H$(Neatness, Hired) = 2.4261 &   &          \\ \hline
		$SU$(Neatness, Hired) = 0.0535 &   &         
		    \\ \hline
	\end{tabular}
	\caption{Calculating $SU$(Neatness, Hired) from Table \ref{tab:msu4b}.}
	\label{SU_NeatnessCalculation}
\end{table}

The full set of resulting $SU$ values between each of the variables and the \textit{Hired} response is:
\begin{itemize}
	\item $SU$(Neatness, Hired) = 0.0535
	\item $SU$(Creativity, Hired) = 0.4627 (largest computed $SU$)
	\item $SU$(Punctuality, Hired) = 0.0535
	\item $SU$(IQuotient, Hired) = 0.0535
	\item $SU$(Attention Type, Hired) = 0.0192
\end{itemize}

Out of many possible predictive models, a model to predict whether a student gets hired based on a single predictor variable is perhaps the simplest one. Of all the 5 feature vs class pairs in the data, the \textit{creativity-gotHired} pair has the highest SU value. Let us analyze this a bit more.

Suppose we need to group our 20 cases by each category of Creativity and by each category of GotHired. When we do this, Creativity and GotHired produce about the same groupings, as seen in Table \ref{tab:CreatGotHired}: creativity D is associated with GotHired=Y, and creativity levels S or I are associated with GotHired=N.

\begin{table}[ht]
	\centering
	\small
	\setlength\arrayrulewidth{0.1pt}
	\renewcommand{\arraystretch}{1.5}
	
	\begin{tabular}{|l| c c c|}
		\hline
		\multirow{2}{*}{--}& \multirow{2}{*}{Creativity,} & \multirow{2}{*}{Creativity,} &
		\multirow{2}{*}{Creativity,}  \\
		& D & S & I \\
		\hline
		\multirow{2}{*}{GotHired, Y}&
		\multirow{2}{*}{$8$} & \multirow{2}{*}{$1$} &
		\multirow{2}{*}{$0$}  \\
		&  &  & \\[-10pt]
		\multirow{2}{*}{GotHired, N}& \multirow{2}{*}{$1$} & \multirow{2}{*}{$4$} &
		\multirow{2}{*}{$6$}  \\
		& & & \\
		\hline
	\end{tabular}
	\caption{\label{tab:CreatGotHired}{Counts for Creativity and Getting Hired}}
\end{table}

Thus, creativity is akin to getting hired, or \textit{Creativity} is similar to \textit{GotHired} in the correlation sense. From the point of view of a two-way model, if we know that the value of creativity is D, we can almost predict that the value of GotHired will be Y, while for creativities S or I the candidate probably will not be hired.

The same kind of analysis could be done for the other four resulting SU values, detecting similarities with GotHired in accordance with each SU value.

Note that, if we replace \textit{GotHired} by its negative \textit{NotHired}, all the $SU$ values would remain the same as only the probabilities of the labels (not labels themselves) participate in $SU$ computations.

\subsection{Entropic equivalence and the metric structure}

In many situations, given a set $\Omega$, it is useful to have a function $f\colon \Omega\times \Omega\to \mathbb{R}$ which tells us how similar or dissimilar any given pair of elements of $\Omega$ is. One such function is the well-known {\it distance metric}, which is any function $d\colon \Omega \times \Omega\to \mathbb{R}$ satisfying the following conditions:

\begin{enumerate}
	\item $d(x, y) \ge 0$ (non-negativity),
	\item $d(x, y) = d(y, x)$ (symmetry),
	\item $d(x, z) \le d(x, y) + d(y, z)$ (triangle inequality),
	\item $d(x, y) = 0$ if and only if $x = y$. (identity of indiscernibles)
\end{enumerate}

Less standardized is the concept of a {\it similarity metric}. Here, we recall the definition given by Chen et al. in \cite{TCS:Chen-2009}. 

\begin{definition} 
	Given a set $\Omega$, a real-valued function $s(x, y)$ defined on the Cartesian product $\Omega \times \Omega$ is a similarity metric if, for any $x, y, z \in \Omega$, it satisfies the following conditions:
	\begin{enumerate}
		\item $s(x, y) = s(y, x)$ (symmetry)
		\item $s(x, x) \ge 0$ (non-negativity)
		\item $s(x, x) \ge s(x, y)$ (self-identity)
		\item $s(x, y) + s(y, z) \le s(x, z) + s(y, y)$ (triangle inequality)
		\item $s(x, x) = s(y, y) = s(x, y)$ if and only if $x = y$. (identity of indiscernibles)
	\end{enumerate}
\end{definition}

In \cite{TCS:Chen-2009}, the authors also gave conditions for a similarity metric to yield a distance metric. More concretely, they established a precise relationship between normalized distances and similarity metrics. Here, {\it normalized} means that both the distance and the similarity map into the interval $[0, 1]\subset \mathbb{R}$. For the reader’s convenience, let us recall their result here.

\begin{lemma}[\!\!{\cite[Corollary 1]{TCS:Chen-2009}}]\label{lem:similarity-distance}
	If $s(x, y)$ is a normalized similarity metric and for any $x$, $s(x, x) = 1$, then $\frac{1}{2}
	(1 - s(x, y))$ is a normalized distance
	metric. If, in addition, $s(x, y) \ge 0$, then $1 - s(x, y)$ is a normalized distance metric. If $d(x, y)$ is a normalized distance metric, then
	$1 - d(x, y)$ is a normalized similarity metric.
\end{lemma}

Looking at pairs of categorical variables, we know that the more correlated two variables are, the greater their SU value, hence SU appears to be a similarity measure. Showing that SU is a similarity metric requires the concept of {\it indiscernibility} defined for categorical random variables.
Let us make this observation precise.

\begin{definition} \label{defIndiscernibles}
	Let $\cal C$ be a set of categorical random variables defined on a sample space $D$. If $X\colon D\to \Sigma_X$ and $Y\colon D\to \Sigma_Y$ are two elements in ${\cal C}$, we say that they are \textbf{indiscernible} if there is a bijection \(h\colon \Sigma_X\to \Sigma_Y\) such that \(Y=h\circ X\) almost everywhere.  
\end{definition}

Let us notice that, for two indiscernible categorical random variables \(X\) and \(Y\), their corresponding partition \({\cal X}\) and \({\cal Y}\) are equal almost everywhere. We will denote this by \({\cal X}\sim {\cal Y}\). In particular, the relative entropy given by Equation~\eqref{eq:relative-entropy-partition} is zero for two indiscernible categorical random variables. 

The indiscernibility of categorical random variables introduces an equivalence relation in $\cal C$. Let us denote by $\overline{{\cal C}}$ the space of equivalence classes. Notice that any two equivalence classes may have the same number of categories or a different number of categories.

\begin{example}[of indiscernibles] 
	Consider variable $X$ with values 1, 2 or 3 and variable $Y$ with values A, B or C. Table \ref{tab:adataset2} shows a sample of 10 individuals. Considering the probability given by the observed frequency, we can see that the histogram of \(X\) is 0.5, 0.1, 0.4, while the histogram of \(Y\) is 0.4, 0.1, 0.5. Here, we can define the bijection \(h\colon \{1, 2, 3\}\to \{A, B, C\}\) given by \(h(1) = C; h(2)=B; h(3)=A\), and then the partition defined on the domain by \(Y\) and \(h\circ X\) are the same.
	
	\begin{table}[ht]
		\centering
		\setlength\arrayrulewidth{0.1pt}
		\renewcommand{\arraystretch}{1.3}
		
		\begin{tabular}{|c c c|}
			\hline
			$Row$ & $X$ & $Y$ \\
			\hline
			1 & 2 & B \\
			2 & 1 & C \\
			3 & 3 & A \\
			4 & 3 & A \\
			5 & 1 & C \\
			6 & 1 & C \\
			7 & 3 & A \\
			8 & 1 & C \\
			9 & 3 & A \\
			10 & 1 & C \\
			\hline
		\end{tabular}
		\caption{\label{tab:adataset2} \textbf{Example of Indiscernibles.}}
	\end{table}
	
\end{example}
Let us observe that the $SU$ can be promoted to the quotient space $\overline{\cal{C}}$ in the natural way, namely, if $[X], [Y]\in {\cal C}$, then
\begin{equation}\label{eq:SU-quotient}
	SU([X], [Y]) := SU(X, Y)
\end{equation}
is well defined. 

Now we can state the following result, where, for simplicity, we will omit the bracket notation for equivalence classes.

\begin{theorem} \label{TheoSimilarityMetric}
	Given a space $\mathcal{C}$ of categorical random variables, the SU as defined in equation \eqref{suDefinition}, induces via \eqref{eq:SU-quotient} a positive normalized similarity metric in the space $\overline{\cal{C}}$. That is, for any $X$, $Y$ and $Z$ in $\overline{\cal{C}}$, the following conditions hold:
	\begin{enumerate}
		\item Symmetry: $SU(X, Y) = SU(Y, X),$
		\item Reflexivity: $SU(X, X) \geq 0,$
		\item Self-similarity: $SU(X, X) \geq SU(X, Y),$
		\item Triangle inequality: $SU(X, Y) + SU(Y, Z) \leq SU(X, Z) + SU(Y, Y),$
		\item Identity of indiscernibles: $SU(X, X) = SU(Y, Y) = SU(X, Y)$ if and only if $X = Y$.
		\item Positivity and normalization: $SU(X, Y)\in [0, 1]$.
	\end{enumerate}
\end{theorem}

\begin{proof}
	First, note that if $X$, $Y$ and $Z$ are any categorical random variables, SU as defined in equation (\ref{suDefinition}) is a composite random variable, and at the same time it is a real-valued function on the Cartesian product $\mathcal{C} \times \mathcal{C}$.
	
	({\bf Condition 1}) is straightforward by substitution into equation (\ref{suDefinition}). Likewise, for the reflexivity ({\bf Condition 2}), using equation (\ref{suDefinition}) we get the self-similarity $SU(X, X) = 1 \geq 0$.
	
	The self-similarity ({\bf Condition 3}) establishes that $SU(X, X)$ is greater than or equal to any other value of SU(X, Y). This is indeed true as we have just seen that $SU(X, X) = 1$ which is the maximum value achievable by the measure.
	
	Equivalent to the triangle inequality in a distance metric, the triangle inequality (\textbf{Condition 4}) states that the similarity between $X$ and $Z$ through $Y$ is no greater than the direct similarity between $X$ and $Z$ plus the self-similarity of $Y$. Let us start by rewriting the inequalities for the two-variable cases, using entropic ratio notation (see definition in Lemma \ref{lem:entropic-ratio}). For the variables $X$, $Y$, and $Z$, we get:
	\begin{align} 
		\frac{1}{2} &\le R(X, Y) \le 1,  \nonumber\\
		\frac{1}{2} &\le R(Y, Z) \le 1,  \nonumber\\
		\frac{1}{2} &\le R(X, Z) \le 1.  \nonumber
	\end{align}
	
	Also note that $SU(X, Y) = 2(1 - R(X,Y))$. We now show Condition 4 in four steps.
	
	a. Working with corresponding sides, add the first two inequalities above and then subtract the third,
	\begin{align} 
		1 &\le R(X,Y) + R(Y,Z) \le 2  \\
		1/2 &\le R(X,Y) + R(Y,Z) - R(X,Z) \le 1
	\end{align} 
	
	b. From here on, use only the left and the middle members. Transpose two terms to the left, getting
	\begin{align} 
		1/2 - R(X,Y) - R(Y,Z) &\le -R(X,Z)
	\end{align} 
	
	c. Add (1 + {\textonehalf}) to each side, then multiply each side by 2.
	\begin{align}
		1 - R(X,Y) + 1 - R(Y,Z) &\le 1 - R(X,Z) + (1/2) \\
		2(1 - R(X,Y)) + 2(1 - R(Y,Z)) &\le 2(1 - R(X,Z)) + 2(1/2)
	\end{align}
	
	d. Remembering that $SU(Y,Y)$ = 1,
	\begin{align}
		SU(X, Y) + SU(Y,Z) &\le SU(X,Z) + 1  \\
		SU(X, Y) + SU(Y,Z) &\le SU(X,Z) + SU(Y,Y)
	\end{align}
	
	(\textbf{Condition 5}) results from applying $SU(X,X)$ = 1 once again, and noting that $X=Y$ is understood as equality in entropy.
	
	Finally, notice that (\textbf{Condition 6}) was already established in Lemma \ref{lem:entropic-ratio}. This finalizes the proof.
\end{proof} 

As a direct consequence of this theorem, together with Lemma \ref{lem:similarity-distance}, we have the main result of this section.

\begin{theorem}
	Given a space $\cal C$ of categorical random variables, let $\overline{\cal C}$ be its quotient space defined by indiscernibility. The quantity $1 - SU(X, Y)$ is a normalized distance metric in $\overline{\cal C}$.
\end{theorem}

Let us notice that, since we are working with categorical variables, it may seem that the metric topology just introduced is discrete. However, this is not the case, as we show in the following theorem:

\begin{theorem}\label{thm:non-discrete-topology}
	The metric topology introduced in $\overline{\cal C}$ is not discrete.
\end{theorem}

\begin{proof}
	Let \(X\) be a categorical random variable, and let us consider \(Y\) as a noisy copy of \(X\), namely: \(Y=X\) with probability \(1-\epsilon\), and an independent redraw of the categories with probability \(\epsilon\). Thus, \(Y\) is almost a function of \(X\), which means that knowing \(X\) almost determines \(Y\). Then, as \(\epsilon \to 0\) we have \(H(Y|X)\to 0\) and \(H(Y)\to H(X)\), hence
	\[SU(X, Y) = 2\frac{H(Y)-H(Y|X)}{H(X)+H(Y)}\to 2\frac{H(X)}{H(X)+H(X)}=1.\]
	It follows that \(d(X, Y) = 1-SU(X, Y)\to 0.\) 
\end{proof}
\section{Algebraic Structure of Categorical Variables} \label{sec:monoidcontinuity}

Let $A$ and $B$ be two categorical random variables defined over a common sample space $D$. To fix the idea, let us assume that $A$ stands for income, which we consider to have three possible values, namely $\{l,m, h\}$ for ``low'', ``middle'' and ``high'', and that $B$ stands for owning a house or not, which we consider to have two possible values, namely $\{y, n\}$, for ``yes'' and ``no''. It is quite common to produce, out of $A$ and $B$, another random variable $C$, which we will denote as $C=A*B$, by the following rule: Given any $p\in D$, then $C(p):= (A(p), B(p))$.

Thus, if $p\in D$ is an individual with middle income, and owns some house, then $C(p) = (m, y)$. Let us observe that the possible values for $C$ are the set
\[\{(l, y), (l, n), (m, y), (m, n), (h, y), (h, n)\}.\]

If we denote by $\Sigma_A$ the codomain of the random variable $A$ and by $\Sigma_B$ the codomain of the random variable $B$, then we get that the codomain of $C$ is $\Sigma_C=\Sigma_A\times \Sigma_B$. Let us also recall that, if $P$ is the probability measure in the sample space $D$, and $P_A$, $P_B$ are the probability mass functions of $A$ and $B$, respectively, then, in the measurable product space $\Sigma_A\times \Sigma_B$ we have the joint probability mass function $P_A\times P_B\colon \Sigma_A\times \Sigma_B\to \mathbb{R}$ given by
\[P_A\times P_B(x, y) = P(A=x\wedge B=y) = P(A=x|B=y)\cdot P(B=y).\]

Thus, if we denote by ${\cal C}$ the set of all categorical random variables defined over $D$, what the previous procedure yields is an internal law of composition in ${\cal C}$.

In this section we establish that this operation introduces in ${\cal C}$ an algebraic structure (more specifically, in the quotient space $\overline{\cal C}$ given by the indiscernibility equivalence relation), and explore some of its basic properties.  

\subsection{Formal definitions}

Let us fix again the space ${\cal C}$, and recall that it consists of all categorical random variables defined on the sample space $D$. Let us also denote again by $\overline{\cal C}$ the quotient space given by the indiscernibility equivalence relation in ${\cal C}$.

\begin{definition}
	Given $A\colon D\to \Sigma_A$ and $B\colon D\to\Sigma_B$, two categorical variables in ${\cal C}$. Then we define a new categorical random variable
	\[C\colon D\to \Sigma_C=\Sigma_A\times \Sigma_B; C(p) := (A(p), B(p)).\]
	Let us denote this by $C = A*B$, and call it the {\bf joint} of $A$ and $B$.
\end{definition}
Notice that, in the context of the previous definition, for each \((a, b)\in \Sigma_A\times \Sigma_B\) we have
\[C^{-1}(a, b) = A^{-1}(a)\cap B^{-1}(b),\]
hence, if \({\cal A}\) and \({\cal B}\) are the partitions given by \(A\) and \(B\), then the partition given by \(C\) is
\[{\cal C} = {\cal A}\cap {\cal B}.\]

\begin{proposition}
	Given $A, B\in {\cal C}$, the operation
	\[[A]*[B]:= [A*B]\]
	is well defined.
\end{proposition}

\begin{proof}
	Indeed, let $A'\colon D\to \Sigma_{A'}$ and $B'\colon D\to\Sigma_{B'}$ be two categorical random variables equivalent to $A$ and $B$, respectively. Then, the bijections \[h_A\colon \Sigma_{A'}\to \Sigma_A \quad \text{and}\quad h_B\colon \Sigma_{B'}\to \Sigma_{B}\] yield a bijection $h\colon \Sigma_{A'}\times \Sigma_{B'}\to \Sigma_{A}\times \Sigma_{B}$. Also, for the corresponding partitions we have \({\cal A}\sim {\cal A'}\) and \({\cal B}\sim{\cal B'}\). It follows that for \(C = A*B\) and \(C'= A'*B'\) we have \({\cal C}\sim {\cal C'}\), and thus \([C] = [C']\).
\end{proof}

\begin{proposition}
	The joint operation $*\colon \overline{\cal C}\times \overline{\cal C}\to \overline{\cal C}$ is associative and commutative.
\end{proposition}

\begin{proof}
	Consider $[A], [B], [C]\in \overline{\cal C}$. For any representatives,      
	\[A\colon D\to \Sigma_A;\quad B\colon D\to \Sigma_B;\quad C\colon D\to \Sigma_C,\]
	we need to show that  
	
	\[A*(B*C)\colon D\to \Sigma_A\times (\Sigma_B\times \Sigma_C) = \Sigma_A\times \Sigma_B\times\Sigma_C\]
	and
	\[(A*B)*C\colon D\to (\Sigma_A\times \Sigma_B)\times \Sigma_C = \Sigma_A\times \Sigma_B\times\Sigma_C,\]
	are indiscernible.
	
	We clearly have the bijection between the codomains, and also for the corresponding partition we have
	
	\[{\cal A}\cap({\cal B}\cap {\cal C}) = ({\cal A}\cap {\cal B})\cap {\cal C}.\]

	For the commutativity, observe that
    \[
        A*B\colon D\to \Sigma_{A*B} = \Sigma_{A}\times \Sigma_B
        \quad \text{and} \quad 
        B*A\colon D\to \Sigma_{B*A} =\Sigma_B\times \Sigma_A.
    \]
    Hence, we have a bijection between $\Sigma_{A*B}$ and $\Sigma_{B*A}$, and for the partitions we clearly have
	\({\cal A}\cap {\cal B} = {\cal B}\cap {\cal A}.\)
	Thus, we have \([A]*[B] = [B]*[A]\).
\end{proof}

So far, we have an internal operation in the space $\overline{\cal C}$, which is associative and commutative. Next we identify a categorical variable (or better, an equivalence class of categorical variables) that plays the role of neutral element.
\begin{proposition}
	Let $\Phi\colon D\to \{\phi\}$ be a categorical random variable whose only possible outcome is the singleton $\{\phi\}$. Then, for any other random variable $A\colon D\to \Sigma_A$, we have
	\[[A]*[\Phi] = [A].\]
\end{proposition}

\begin{proof}
	Let us first observe that any random variable on $D$ having a singleton as the only possible outcome is equivalent to $\Phi$.  
	Now, we observe that $\Sigma_{A*\Phi} = \Sigma_A\times\{\phi\}$, so we have a bijection $\Sigma_A\to \Sigma_{A*\Phi}$ given by $x\mapsto (x, \phi)$. 
	Also, the partition corresponding to \(\Phi\) is the trivial one \({\cal D} = \{D\}\), and thus the partition of \(A*\Phi\) is
	\({\cal A}\cap{\cal D} = {\cal A}.\)
	Hence, we have \([A]*[\Phi] = [A],\) as claimed.
\end{proof}

Let us summarize the previous discussion in the following theorem:

\begin{theorem}
	The joint operation on $\overline{{\cal C}}$ yields a commutative monoid structure.
\end{theorem}

\subsection{Compatibility of the algebraic and topological structures}
Let us recall that, in the previous section, we found a metric topology structure in the set $\overline{{\cal C}}$, given by the metric $d = 1-SU$. Now we will show that the algebraic structure given on the same set by the joint operation and the topological structure in ${\overline{\cal C}}$ are compatible. 

To do so, let us first rewrite the distance function on a more symmetric form. For that, we recall the symmetry of mutual information, namely:
\[MI(X;Y) = H(X)-H(X|Y) = H(Y)-H(Y|X) = MI(Y;X).\]
Hence, for the distance $d=1-SU$ we get
\begin{equation}
	\begin{split}
		d(X, Y) = 1-SU(X, Y) &= 1-\frac{2MI(X;Y)}{H(X)+H(Y)}\\
		&= \frac{H(X)+H(Y)-H(X)+H(X|Y) - H(Y)+H(Y|X)}{H(X)+H(Y)}\\
		& = \frac{H(X|Y)+H(Y|X)}{H(X)+H(Y)}
	\end{split}
\end{equation}

Also, observe that, given categorical random variables $X\colon D\to \Sigma_X$ and $Y\colon D\to \Sigma_Y$, we have
\[H(X, Y) = H(X*Y) = H({\cal X}\lor {\cal Y}),\]
where ${\cal X}$ and ${\cal Y}$ are the partitions of $D$ associated to the variables $X$ and $Y$, respectively.

Now we have all the tools in place to prove the aforementioned compatibility.
\begin{theorem}
	The joint operation 
	\[*\colon \overline{{\cal C}}\times\overline{{\cal C}}\to \overline{{\cal C}}\]
	is continuous with respect to the topology on ${\overline{{\cal C}}}$ given by $d=1-SU$ and the natural product topology induced on $\overline{{\cal C}}\times\overline{{\cal C}}$.
\end{theorem}

\begin{proof}
	Let $X, Y, Z, W$ be categorical random variables on $D$, for simplicity, we will avoid the equivalence class notation and we will also use the same symbols for their associated partitions. Then we have
	\begin{equation}
		d(X*Y, Z*W) = \frac{H(X*Y|Z*W)+H(Z*W|X*Y)}{H(X*Y)+H(Z*W)}
	\end{equation}
	Applying item $a)$ from Lemma \ref{lemma:relative-entropy-properties}, and the inequality given in equation~\eqref{eq:refinement-inequality} to the terms in the numerator we get:
	\begin{equation}
		\begin{split}
			H(X*Y|Z*W) &= H(X|Z*W)+H(Y|X*Z*W)\leq H(X|Z)+H(Y|W)\\
			H(Z*W|X*Y) &= H(Z|X*Y)+H(W|Z*X*Y)\leq H(Z|X)+H(W|Y).
		\end{split}
	\end{equation}
	
	Thus, we get
	\begin{equation}
		d(X*Y, Z*W) \leq\frac{H(X|Z)+H(Z|X)}{H(X*Y)+H(Z*W)}+\frac{H(Y|W)+H(W|Y)}{H(X*Y)+H(Z*W)}
	\end{equation}
	Now we apply inequality ~\eqref{eq:abs-entropy-less-joint-entropy} to each summand in the denominators, to get
	\begin{equation}
		\begin{split}
			d(X*Y, Z*W) & \leq\frac{H(X|Z)+H(Z|X)}{H(X)+H(Z)}+\frac{H(Y|W)+H(W|Y)}{H(Y)+H(W)}\\
			&= d(X, Z)+d(Y, W)
		\end{split}
	\end{equation}
	Thus, the mapping is contractive, hence, continuous. 
\end{proof}

\section{Concluding Remarks}

In this paper we have demonstrated that $1 - SU$, where $SU$ is the entropic correlation between categorical random variables, is a distance metric, thus endowing the space of categorical random variables with a topological structure. We have also shown that a natural joint operation (denoted $*$) endows the same space with an algebraic structure, concretely, a commutative monoid structure. Finally, we have proven that the topological and algebraic structures are compatible, given that the joint operation is continuous with respect to the topological structure.

The compatibility provides a smooth and straightforward way for the user to interpret and understand the output value from SU. This way, SU becomes more intuitive for statistical practitioners, allowing practitioners to work with and employ this entropic (non-parametric) correlation in much the same way as the Pearson (parametric) correlation has been employed over the years.

Equipped with compatible topology and algebraic structure, $SU$ further formalizes categorical random variables (CRVs) as repositories and proper representations of the concepts conveyed by the variable's categories. Thus, qualitative variables not only can be counted in their categories, but their mutual distances now tell us about similarities or likely associations. These new possibilities and measures are offered in an environment of mathematical rigor.

The authors are now working to extend the findings to the multivariate symmetric uncertainty (MSU) measure on $n$ categorical variables. Note that the MSU measure could as well be called Multivariable Entropic Correlation, as uncertainties get mutually canceled when computing correlation.


{\small
    \begin{thebibliography}{10}
        \bibitem{Bouda}
        J.~Bouda.
        \newblock Lecture 5: Information theory.
        \newblock
          \url{https://www.fi.muni.cz/~xbouda1/teaching/current/IV111/pred-nasky/lecture5.pdf},
          May 2012.
        \newblock Faculty of Informatics.
        
        \bibitem{TCS:Chen-2009}
        S.~Chen, B.~Ma, and K.~Zhang.
        \newblock On the similarity metric and the distance metric.
        \newblock {\em Theoretical Computer Science}, 410(2):2365--2376, 2009.
        
        \bibitem{ref-Fullwood2023}
        J.~Fullwood.
        \newblock An axiomatic characterization of mutual information.
        \newblock {\em Entropy}, 25(4), 2023.
        
        \bibitem{Entropy:Gomez-2021}
        S.~G\'{o}mez-Guerrero, G.~Sosa-Cabrera, E.~Sotto-Riveros, C.~Schaerer, and
          M.~García-Torres.
        \newblock Classifying {D}engue cases using cat{PCA} in combination with the
          {MSU} {C}orrelation.
        \newblock In {\em Proceedings of the {E}ntropy 2021 {C}onference}, 2021.
        
        \bibitem{ref-Gomez2023}
        S.~Gómez-Guerrero.
        \newblock {\em A proposed correlation measure for categorical random
          variables}.
        \newblock PhD thesis, Polytechnic School, Universidad Nacional de Asunción,
          San Lorenzo, Paraguay, 2023.
        
        \bibitem{CUP:Press_et_al:1988}
        W.~H. Press, B.~P. Flannery, S.~A. Teukolsky, and W.~T. Vetterling.
        \newblock {\em Numerical Recipes in C}.
        \newblock Cambridge University Press, Cambridge, UK, 1998.
        
        \bibitem{ML:Quinlan-1986}
        J.~Quinlan.
        \newblock Induction of decision trees.
        \newblock {\em Machine Learning}, 1(1):81--106, 1986.
        
        \bibitem{MKP:Quinlan:1993}
        J.~R. Quinlan.
        \newblock {\em C4.5: Programs for Machine Learning}.
        \newblock Morgan Kaufmann Publishers Inc., San Francisco, CA, USA, 1993.
        
        \bibitem{Renyi-1961}
        A.~Rényi.
        \newblock On measures of entropy and information.
        \newblock In {\em Proceedings of the {F}ourth {B}erkeley {S}ymposium on
          {M}athematics, {S}tatistics and {P}robability}, pages 547--561, 1961.
        
        \bibitem{Shannon-1948}
        C.~Shannon.
        \newblock A mathematical theory of communication.
        \newblock {\em The {B}ell {S}ystem {T}echnical {J}ournal}, 27:379–423,
          623–656, 1948.
        
        \bibitem{SBMAC:Sosa-2018}
        G.~Sosa-Cabrera, M.~Garc\'{i}a-Torres, S.~G\'{o}mez-Guerrero, C.~Schaerer, and
          F.~Divina.
        \newblock Understanding a multivariate semi-metric in the search strategies for
          attributes subset selection.
        \newblock In {\em Proceeding {S}eries of the {B}razilian {S}ociety of
          {C}omputational and {A}pplied {M}athematics}, volume~6, 2018.
        
        \bibitem{IS:Sosa-2018}
        G.~Sosa-Cabrera, M.~Garc\'{i}a-Torres, S.~G\'{o}mez-Guerrero, C.~E. Schaerer,
          and F.~Divina.
        \newblock A multivariate approach to the symmetrical uncertainty measure:
          Application to feature selection problem.
        \newblock {\em Information Sciences}, 494:1--20, 2019.
        
        \bibitem{UG:Suarez-2020}
        J.~L. Suárez, S.~García, and F.~Herrera.
        \newblock A tutorial on distance metric learning: Mathematical foundations,
          algorithms, experimental analysis, prospects and challenges.
        \newblock Technical report, Andalusian Research Institute in Data Science and
          Computational Intelligence - University of Granada, 2020.
        \newblock arXiv:1812.05944.
        
        \bibitem{Book:Viana-Oliveira}
        M.~Viana and K.~Oliveira.
        \newblock {\em Foundations of Ergodic Theory}.
        \newblock Cambridge Studies in Advanced Mathematics, Cambridge University
          Press, 2016.
    \end{thebibliography}
}
\end{document}